\documentclass[11pt]{article}%

\usepackage{graphicx}
\usepackage{amsthm}
\usepackage{amsmath}
\usepackage{amsfonts}
\usepackage{amssymb}
\def\LT{{\mathbb{LT}}}

\def\bft{\mathbf t}
\def\ubft{\underline\bft}
\def\bfsigma{{\mbox{\boldmath${\sigma}$}}}

\newcommand{\N}{{\mathbb{N}}}
\newcommand{\Z}{{\mathbb{Z}}}
\newcommand{\R}{\mathbb{R}}

\newtheorem{theo}{Theorem}
\newtheorem{lm}{Lemma}
\newtheorem{cor}{Corollary}
\newtheorem{prop}{Proposition}

\begin{document}

\begin{center}

{ \large \bf Potts model coupled to causal triangulations}

\vspace{30pt}

{\sl J. Cerda-Hern\'andez}$\,^{a}$ 

\vspace{24pt}

{\footnotesize
$^a$~Institute of Mathematics, Statistics and Scientic Computation,\\ 
University of Campinas - UNICAMP,\\
Rua S\'ergio Buarque de Holanda 651, CEP 13083-859, Campinas, SP, Brazil.\\
E-mail: javier@ime.usp.br.
}

\vspace{48pt}

\end{center}

\begin{abstract}
In this work we study the annealed Potts model coupled to
two dimensional causal triangulations with periodic 
boundary  condition. Using duality on a torus, we provide a relation between the free energy of the Potts model coupled
CTs and its dual. This duality relation follows from the FK representation for the 
Potts model. In order to determine a region where the critical curve for the model can be located we use the 
duality relation and the high-temperature expansion. 
This is done by outlining a region where the infinite-volume Gibbs measure exists and is unique and a region where
the finite-volume Gibbs measure has no weak limit (in fact, does not
exist if the volume is large enough). We also provide lower and upper
bounds for the infinite-volume free energy. 
\end{abstract}


\newpage

\section{Introduction}

Causal triangulation (CT), introduced by Ambj{\o}rn and Loll (see \cite{Ambjorn:1998xu}), together with 
its predecessor  a dynamical triangulation (DT),  constitute attemps to provide a meaning to the formal expressions 
appearing in the path integral quantisation of gravity (see \cite{Ambjorn:1997di}, 
\cite{Ambjorn:2006} for an overview). The idea is to approximate the emerging  geometries by CTs. As a result, we obtain 
a discrete version of the path integral where 
continuum geometries are replaced by a sum over all possible triangulations,  where each triangulation
is weighted by a Boltzmann factor $e^{-\mu |T|}$, with $|T|$ standing for the size of
the triangulation and $\mu$ being the cosmological constant. Then, evaluation of the
partition function is reduced to a purely combinatorial problem that can be solved
utilizing the approach developed in the early work of Tutte \cite{Tutte1962a,Tutte1963} or techniques based 
on random matrix models (see, e.g.,
\cite{DiFrancesco:1993nw}) and bijections to well-labelled trees (see \cite{Schaeffer1997,
bouttier-2002-645}).

Putting a spin system on the collection of all causal triangulations is generally interpreted  as
a coupling gravity with matter, which makes it interesting to study  the
$q$-state  Potts model coupled to CTs from  a physical point of view.
The clasical example  of such model is  the  two-state  Potts model (Ising model) 
coupled to a CT introduced in \cite{Ambjorn:1999gi}. For the Ising model existence of Gibbs measures and phase transitions
has been recently proved  (see \cite{Ambjorn:1999gi}, \cite{Benedetti:2006rv}, 
\cite{HeAnYuZo:2013}, \cite{cerda} and \cite{Napolitano:2015}  for details). For the numerical results for the 3-state Potts model
coupled to CTs we refer the reader to \cite{Ambjorn:2008jg}.
In this work we focus on the $q$-state Potts model coupled to CTs for any $q\geq 2$. Our main goal is to
derive properties of the phase transition for the model and the critical curve by 
defining a region in the quadrant of parameters $\beta,\mu >0$ where the infinite-volume free energy  has a limit, which  
implies uniqueness of the Gibbs measure. In adition, we prove in Corollary \ref{asympt_beh} that the critical curve of the model is 
asymptotic if $\beta$ is large  to $\frac{3}{2}\beta+\ln2$, and  $C(q) +o(\beta^2)$ if  $\beta$ is small, where $C(q)=\ln2\sqrt{q}$ for 
the Potts model coupled on CTs, and $C(q)=\ln2q$ for its dual. 
(see Figure \ref{fig5}). In order to obtain these results we utilize the FK-Potts models, introduced  by Fortuin and 
Kasteleyn (see \cite{FK:1972}). These representations were successfully employed in order to obtain 
important results for Ising and Potts models on the hypercubic lattice. 
Since this representation permits the use of geometric properties of the triangulations, we utilize it to 
derive a duality relation for the parameters of the model and asymptotic behavior of the critical curve.

In general, the FK-Potts model on a finite connected graph (not necessarily planar) is a model of edges of the graphs, where each
edge is either closed or open. The probability of a  given configuration is proportional to 
$$p^{\#\mbox{open edges}} (1-p)^{\#\mbox{closed edges}} q^{^{\#\mbox{clusters}}},$$
where  $p\in[0,1]$ and  the cluster-weight $q\in(0,\infty)$ are the parameters of the model. 
For $q\geq 1$, the model can be 
extended to an infinite graph. In this case the model exhibits a phase transition for some critical parameter $p_c(q)$, which
depends on the geometry of the graph. In the case of planar graphs, there exits a relation between  FK-Potts  models 
on a graph and its dual, with the same parameter $q$ and appropriately related 
parameters $p$ and $p^*=p^*(p)$. For a detailed introduction to the FK-Potts model
we refer the reader to \cite{Grimmett:2006}.
In the case of FK-Potts model defined on a causal triangulation  $\bft$ with periodic boundary 
condition (see Figure \ref{fig1} for a geometric representation), the partition
function of the FK-Potts model on $\bft$  cannot be written exactly as a partition function of a FK-Potts model on its dual $\bft^*$,
however it will be sufficient in order to obtain a duality relation of the 
parameters in the thermodynamic limit. This relation together with 
the Edwars-Sokal coupling, using $p=1-e^{-\beta}$, permits find a relation 
between the parameters $(\beta,\mu)$ of the Potts model coupled to CT and the parameters $(\beta^*,\mu^*)$ of  its dual 
for the infinite-volume (thermodynamic limit).

The paper is organized as follows. In Section \ref{Sect2}, we introduce notation and give a summary of the 
main features of CTs. We also define the annealed Potts model coupled to CTs, and we establish the main results of 
this work, Theorem \ref{theo_duality_main1} and \ref{theo_bounds_main2}. In Section \ref{FK_Potts model}, 
we describe the FK-Potts model and establish the technical Lemma \ref{duality} of duality that we will be used 
in next sections. Section  \ref{Sect3} contains 
the proof of Theorem \ref{theo_duality_main1}. This result 
will play a key role in the proof of Theorem \ref{theo_bounds_main2}. Also, in Corollary \ref{asympt_beh} we provide asymptotic 
behavior for the critical curve. 
In Section  \ref{proof_theo2}, utilizing the High-T expansion for $q$-state Potts model, we prove Theorem \ref{theo_bounds_main2}.
The paper concludes with Section \ref{conn_ising}, where we compare our results with the properties obtained in 
\cite{HeAnYuZo:2013} and \cite{Ambjorn:1999gi}.

\section{Notations and main results}\label{Sect2}

In this section we firts introduce notation and give a summary of the causal dynamical triangulation, 
$q$-state Potts model and we define the Potts model coupled CTs. Finally, we give a short of the Edwards-Sokal coupling. 
We refer to \cite{MYZ2001}, \cite{Grimmett:2006}, \cite{HeAnYuZo:2013}, for more details.
We attempt at establishing  regions  where the infinite-volume free energy converges, yielding results on the 
convergence and asymptotic properties of the partition function and the Gibbs measure. 

\subsection{Two-dimensional Lorentzian models}\label{Sect2.1}

We will work with rooted causal dynamic triangulations of the cylinder 
$C_N = {\mathcal S}\times [0,N]$, $N = 1, 2, \dots$, which 
have $N$ bonds (strips) ${\mathcal S}\times [j,j+1]$.  Here ${\mathcal S}$ stands for
a unit circle (see Figure \ref{fig1} for a geometric representation of a CT). Formally, a triangulation 
$\bft$ of $C_N$ is called a {\it rooted causal dynamic 
triangulation} (CTD) if the following conditions hold:
\begin{itemize}
\item each triangular face of $\bft$ belongs to some strip $\mathcal S \times [j, j + 1]$, $j =
1, \dots, N-1$, and has all vertices and exactly one edge on the boundary
$(\mathcal S \times \{j\}) \cup (\mathcal S\times \{j+1\})$ of the strip $\mathcal S\times [j, j + 1]$;
\item the number of edges on $\mathcal S \times \{j\}$ should be finite for any $j = 0, 1, \dots, N-1$: let  $n^j = n^j(\ubft)$ be  the number of edges on $\mathcal S \times \{j\}$, then 
$1 \leq n^j < \infty$ for all $j = 0, 1, \dots, N-1$.
\end{itemize}
and have a root face, with the anti-clockwise ordering of its vertices $(x,y,z)$, where $x$ and $y$ lie
in ${\mathcal S} \times\{0\}$.

The CTs arise naturally when physicists attempt to define a
fundamental path integral in quantum gravity. See \cite{Ambjorn:1997di} for a review of the relevant literatute.  For 
a rigorous mathematical background of the model we refer to \cite{MYZ2001}. Additional properties of CTs have been studied in
\cite{SYZ1}.

A rooted CT $\bft$  of $C_N$ is identified with a {\it compatible sequence}
$$\bft = (\bft(0), \bft(1), \dots, \bft (N-1)),$$ where $\bft (i)$ is a triangulation of the 
strip ${\mathcal S}\times [i,i+1]$. The compatibility means that 
\begin{equation}\label{compat}
n_{up}(\bft(i+1))=n_{do}(\bft(i)),\quad i= 0,\dots, N-2.
\end{equation}

Note that for any edge lying on the slice $\mathcal S \times \{ i\}$ belongs to exactly two 
triangles: one up-triangle from $\bft (i)$ and one down-triangle from $\bft (i-1)$. This provides 
the following relation: the number of triangles in the triangulation $\bft$, denoted by $n(\bft)$, is twice the total number of 
edges on the slices. More precisely, remind that  $n^i$ is the number of edges on slice $\mathcal S \times \{i\}$. 
Then, for any $i=0,1,\dots, N-1$,
\begin{equation}\label{et2-yamb}
n(\bft(i)) = n_{up} (\bft (i)) + n_{do}(\bft (i)) = n^i + n^{i+1}.
\end{equation}

In the usual physical approach to statistical models, the computation of the partition function is the firts step towards 
a deep understanding of the model, enabling for instance the computation of the free energy and study phase 
transition of the model. Follow this approach, the 
computation of the partition function for the case of pure CTs, was first introduced and computed in \cite{Ambjorn:1998xu} 
(see also \cite{MYZ2001} for a mathematically rigorous account). 

Since we are interested in the bulk behavior of Potts model coupled to CTs, we will in the 
following for simplicity choose triangulations with  periodical
spatial boundary conditions, i.e. the strip $\bft (N-1)$ is compatible with $\bft (0)$, and let $\LT_N$ denote the set of causal 
triangulations on the cylinder $C_N$ with this boundary condition, thus the partition function for rooted CTs in the 
cylinder $C_N$ with periodical
spatial boundary conditions and for the value of the cosmological constant $\mu$ is given by
\begin{equation} \label{yamb-pf1}
Z_N(\mu)=\sum_{\bft} e^{-\mu n(\bft) } = \sum_{(\bft(0), \dots, \bft(N-1))} \exp \Bigl\{-\mu \sum_{i=0}^{N-1} n(\bft(i)) \Bigr\}.
\end{equation}

Moreover, the periodical spatial boundary condition on the CTs permits to write the partition function  
$Z_N(\mu)$ in a trace-related form 
\begin{equation}\label{yamb-pf1-2}
Z_N(\mu)= \mbox{tr}\; \bigl( U^N \bigr).
\end{equation}
This gives rise to a transfer matrix $U=\{u(n,n^\prime)\}_{n,n^\prime = 1, 2, \dots}$
describing the transition from one spatial strip to the next one. It is an infinite matrix with
positive entries
\begin{equation}\label{yamb-tmpg}
    u(n,n^\prime) = \binom{n+n^\prime-1}{n-1}e^{-\mu(n+n^\prime)}.
\end{equation}

\begin{figure}[t]
\begin{center}
\includegraphics[height=6cm,width=12cm]{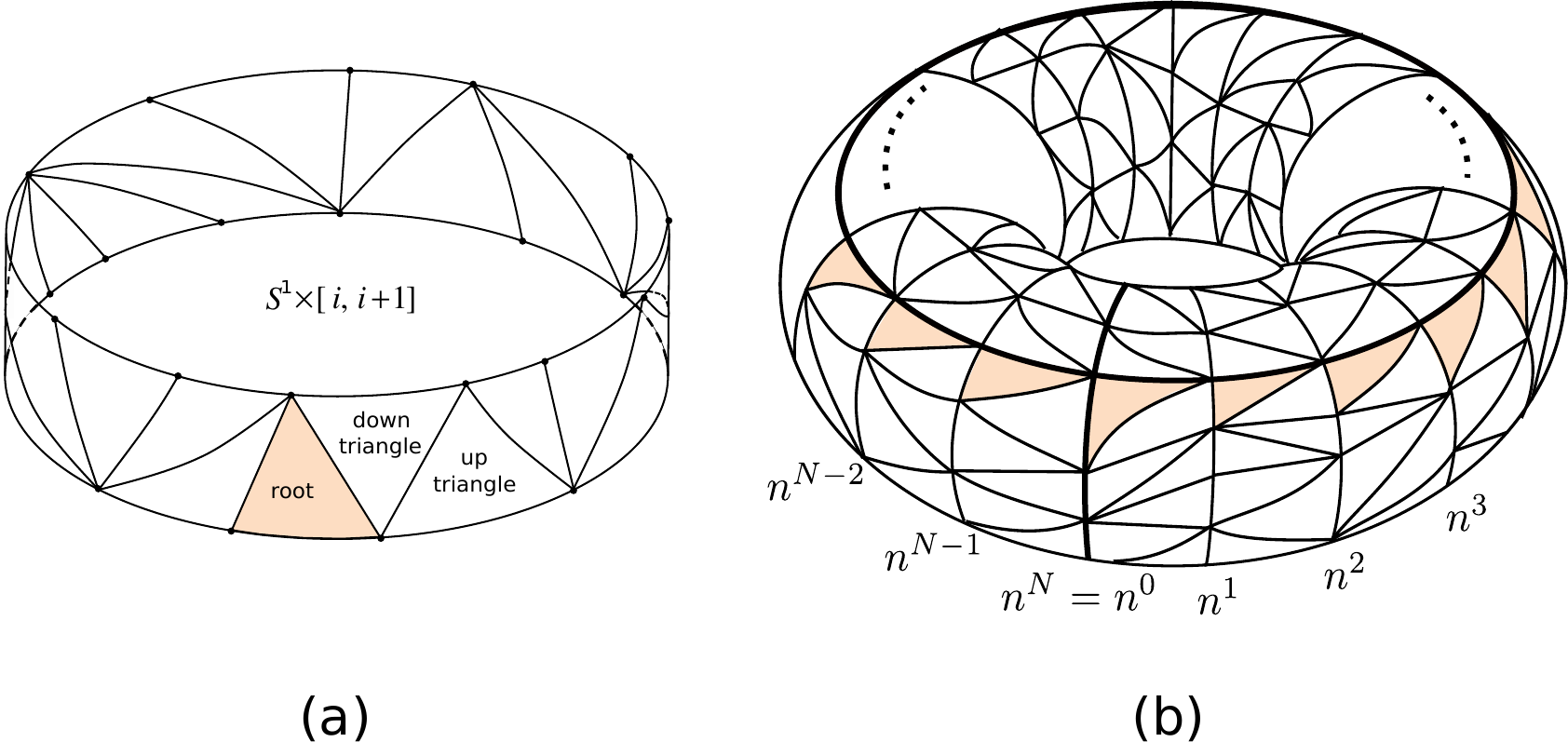}
\end{center}
\caption{(a) A strip triangulation of ${\mathcal S}\times [i,i+1]$. (b) Geometric representation of a CT
with periodic spatial boundary condition.}%
\label{fig1}
\end{figure}

Employing  the $N$-strip partition function for pure CTs with periodical boundary condition, defined by the formula 
(\ref{yamb-pf1}), we define the  $N$-strip Gibbs probability distribution for pure CTs
\begin{equation}\label{Q}
 \mathbb{Q}_{N,\mu}(\bft)=\frac{1}{Z_N(\mu)} \mbox{e}^{-\mu n(\bft)}.
\end{equation}
In 2001, the paper \cite{MYZ2001} computed the partition function of the model and proved  existence of a weak limit of 
the measures $\mathbb{Q}_{N,\mu}$ for $\mu\geq \ln2$,  using transfer-matrix formalism 
and tree parametrization of Lorentzian triangulations (CTs). The weak limit measure  permits describe each triangulation
via  a positive recurrent Markov chain in the subcritical case $\mu>\ln2$, and as the branching process with
geometric offspring distribution with parameter $1/2$, conditioned to non-extinction at infinite in the critical 
case, $\mu=\ln2$ (see \cite{MYZ2001} for more details).

The transfer-matrix formalism suggests  that, as $N\to \infty$, the partition
function is controlled by the largest eigenvalue $\Lambda$ of the transfer matrix (\ref{yamb-tmpg}):
\begin{equation}\label{aproxLambda}
Z_N(\mu) = {\rm tr}\;U^N \sim\Lambda^N,
\end{equation}
where
\begin{equation}\label{Lambda(g)}
\Lambda:= \Lambda(\mu)=\left[ \frac{1- \sqrt{1-4\exp(-2\mu)}}{2\exp(-\mu)}\right]^2.
\end{equation}
That heuristic result was proved in \cite{HeAnYuZo:2013}.
The following properties hold and will be utilized  to  prove the main results.

\noindent {\bf Property 1.} {\rm{(Theorem 1 in \cite{MYZ2001}).}} {\it For any $\mu > \ln 2$ 
the following relation holds true:
\begin{equation}\label{yamb-e13}
\lim_{N\to\infty}\frac{1}{N}\ln\,Z_N(\mu)=\ln\,\Lambda (\mu).
\end{equation}
Further, the $N$-strip Gibbs measure $\mathbb{Q}_{N,\mu}$ converges 
weakly to a limiting measure $\mathbb{Q}_{\mu}$.}

\noindent {\bf Property 2.} {\rm{(Proposition 5, \cite{MYZ2001}).}} {\it For any 
$\mu <\ln 2$, the $N$-strip partition function $Z_N(\mu)$ for pure CTs  exists only if 
\begin{equation}\label{yamb-e14}
\mu > \ln \left( 2\cos \displaystyle\frac{\pi}{N+1} \right).
\end{equation}}
Another proof  of Property 1 can be found in \cite{HeAnYuZo:2013}. In order to prove this property the authors 
utilize the transfer-matrix formalism and Krein-Rutman theorem. Inequality (\ref{yamb-e14}) in Property 2 implies that
if $\mu<\ln2$, then there exists  $N_0\in\N$ such that  $Z_N(\mu)=\infty$ if $N>N_0$.

\subsection{Potts model coupled to CT} \label{sect2.2}

Let $\bft$ be a CT on the cylinder $C_N$ with periodic boundary condition. Each triangulation 
$\bft$ can be view as a graph $\bft=(V(\bft),E(\bft))$ embedded on a torus. 
Potts spin systems are generalizations of the Ising model. Whereas
in Ising systems the spins on two different values, in the $q$-state Potts model 
$q$ distinct values, represent by the elements of the set $\{1,\dots,q\}$,  are allowed on any 
vertex from the triangulation  $\bft$.  We consider the product sample space $\Omega(\bft)=\{1,\dots,q\}^{V(\bft)}$ and  
we consider a usual (ferromagnetic) $q$-state Potts model energy, where two spins
$\sigma (t)$ and
$\sigma (t')$ interact if their supporting vertices $t$, $t'$
are connected by an common edge; such vertices are called nearest neighbors, and this
property is reflected in the notation $\langle t, t'\rangle$.
Thus, the Hamiltonian used for the  $q$-state Potts  model on $\bft$ is given by
\begin{equation}\label{hamilton}
{\mathbf h}(\bfsigma )=-\sum_{\langle t,t' \rangle}\delta_ {\sigma (t),\sigma (t')}.
\end{equation}
The partition function for the $q$-state Potts model on $\bft$ is define by 
\begin{equation}
Z_{P}(\beta,q,\bft) = \sum_{\bfsigma} \exp\Bigl\{-\beta\mathbf{h}(\bfsigma) \Bigr\},
\end{equation}
where the summation is over any configurations $\bfsigma\in \{1,\dots,q\}^{V(\bft)}$. 
Thus, the $q$-state Potts measure on $\bft$ is define as follows
\begin{equation}
\mu^{\bft}_{\beta,q} (\bfsigma)= \frac{1}{Z_{P}(\beta,q,\bft)}
\exp\Bigl\{-\beta\mathbf{h}(\bfsigma) \Bigr\}.
\end{equation}
Using the partition
function  for the $q$-state Potts model on a fixed $\bft$, we define the partition function for the annealed $q$-state Potts 
model  coupled to CTs, at inverse temperature $\beta >0$ and cosmological constant $\mu$, as follows
\begin{equation}\label{yamb-pf}
\Xi_N(\beta,\mu)=\sum_{\bft} \exp\Bigl\{ -\mu n(\bft) \Bigr\} Z_{P}(\beta,q,\bft)
\end{equation}
where $n(\bft)$ stands for the number of triangles in the triangulation $\bft$. Similarly, 
we  introduce the $N$-strip Gibbs probability distribution
associated with (\ref{yamb-pf})
\begin{eqnarray}\label{yamb-Gd}
\mathbb P^{\beta,\mu}_N (\bft,\bfsigma) &=& \frac{1}{\Xi_N (\beta,\mu )} \exp\Bigl\{-\mu n(\bft) -\beta {\bf h}(\bfsigma)\Bigr\},
\end{eqnarray}
and we denote by $\mathcal{G}_{\beta,\mu}$ the set of {\it Gibbs measures} given by the closed convex hull of the set
of weak limits:
\begin{equation}
 \mathbb P^{\beta,\mu}=\lim_{N\to\infty} \mathbb P^{\beta,\mu}_N,
\end{equation}

In general, the $q$-state  Potts model  can be defined  on a general lattice $G$. Therefore, it is possible define the 
$q$-state  Potts model on the dual $\bft^*$ of the triangulation $\bft$ (see next section for a formal definition 
of $\bft^*$). The 
partition function for the $q$-state  Potts model on $\bft^*$ will denote by  $Z_{P}(\beta^*,q,\bft^*)$. 
Finally, we define the partition function for the $q$-state Potts model  coupled to dual CTs  $\Xi_N^*(\beta^*,\mu^*)$
as follow
\begin{equation}\label{yamb-pf_dual}
\Xi_N^*(\beta^*,\mu^*)=\sum_{\bft} \exp\Bigl\{ -\mu^* n(\bft) \Bigr\} Z_{P}(\beta^*,q,\bft^*).
\end{equation}

\subsection{Main results}\label{main_results}
In  the present article, we prove the following duality relation.

\begin{theo}\label{theo_duality_main1}
Let $q\geq 2$. The free energy of the $q$-state Potts model coupled to causal triangulation and its dual satisfied 
the following duality relation
\begin{equation}\label{eq2.31}
\lim_{N\to\infty}\frac{1}{N} \ln \Xi_N(\beta,\mu) =\lim_{N\to\infty} \frac{1}{N} \ln \Xi_N^*(\beta^*, \mu^*)
\end{equation}
where  $\Xi_N$, $\Xi_N^*$ denote the partition function of the  $q$-state Potts model coupled to CT and coupled dual CT 
respectively (defined in Section \ref{sect2.2}), and 
\begin{equation}\label{duality_relation}
\beta^*= \ln\left( 1+ \frac{q}{e^\beta -1} \right), \quad \mu^* = \mu - \frac{3}{2}\ln(e^\beta -1) + \ln q.
\end{equation}
\end{theo}
Thus, equation (\ref{eq2.31}) relates the free energy of the $q$-state
Potts model coupled to CTs  with
the free energy of its dual, and maps the high and low 
temperature  of the dual models onto each other.
We will use the duality relation of Theorem \ref{theo_duality_main1} and the high-temperature expansion for the 
$q$-state Potts model  for determine a region in the
quadrant of parameters where the critical curve for the $q$-state Potts model
coupled CTs and its dual can be located (see Figure \ref{fig5}). 

\begin{figure}[t!]
\begin{center}
\includegraphics[width=12.5cm]{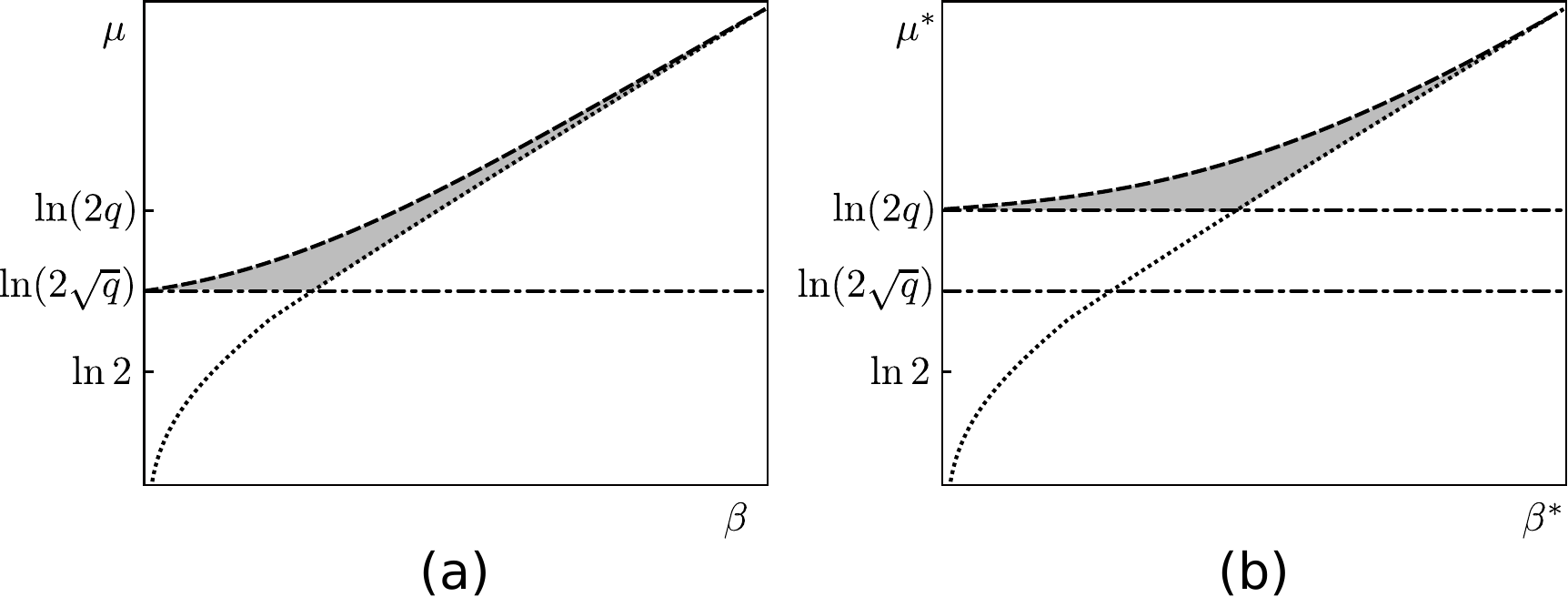}
\end{center}
\caption{Illustrating the region where the critical curve for Potts model coupled CTs and its dual can be located.}
\label{fig5}
\end{figure}

Understanding by {\it critical curve} of the model the boundary of the
domain of parameters $\beta$ and $\mu$ ($\beta^*$ and $\mu^*$ on its dual, respectively) where the model exhibits 
subcritical behavior, this paper makes a rigorous derivation of the subcriticality domain for an $q$-Potts model
coupled to two-dimensional CT and a domain where the tipical infinite-volume Gibbs measure  there no exists. The proof
involve two techniques: the duality relation, Theorem \ref{theo_duality_main1}, and high-temperature expansion for the 
$q$-state Potts model.  In Figure \ref{fig5}, we show the region where the critical curve of the model should 
be located (gray region), and this figure show that critical curve for the model is asymptotic to $\frac{3}{2}\beta+\ln2$.\\

\noindent Define the sets
$$\begin{array}{ccl}
 {\bf \Sigma} &=& \left\{(\beta,\mu)\in\mathbb{R}^2_{+} : \mu <  \max\left\{ \ln(2\sqrt{q}), 
\displaystyle\frac{3}{2}\ln\left(\mbox{e}^{\beta} -1\right) + \ln2\right \}  \right\},
\end{array}$$ 
and 
$$\begin{array}{ccl}
 {\bf \Sigma^*} &=& \left\{(\beta^*,\mu^*)\in\mathbb{R}^2_{+} : \mu^* <  \max\left\{ \ln(2q), 
\displaystyle\frac{3}{2}\ln\left(\mbox{e}^{\beta^*} -1\right) + \ln2\right \}  \right\}.
\end{array}$$ 
We prove the following theorem for existence and no existence of Gibbs measure for the 
model.
\begin{theo}\label{theo_bounds_main2}
Let $q\geq 2$. 
\begin{enumerate}
\item[$(a)$] {\it Potts model coupled to CTs.}  If $(\beta,\mu)\in {\bf \Sigma}$ then there exists $N_0\in\N$ such that  the partition function $\Xi_N(\beta,\mu)=+\infty$ whenever 
$N > N_0$. Moreover, the Gibbs distribution  $\mathbb{P}^{\beta,\mu}_N$ with periodic boundary conditions cannot be defined 
by using the standard formula with  $\Xi_N(\beta,\mu)$ as a normalising denominator, consequently, there is 
no limiting probability measure $\mathbb{P}^{\beta,\mu}$ as $N\to\infty$.  Furthermore, if $(\beta,\mu)$ satisfied 
\begin{equation}
\mu > \displaystyle\frac{3}{2} \ln\left( q +e^\beta -1\right) + \ln 2
-\ln q +\displaystyle\frac{3}{2}\ln\left(1 + (q^{2/3} -1)\displaystyle\frac{e^\beta -1}{q+e^\beta -1} \right), 
\end{equation} 
the infinite-volume free energy exists, i.e. the following limit there exists:
$$ \lim_{N\to\infty} \displaystyle\frac{1}{N}\ln \Xi_N(\beta,\mu).$$
Moreover, as $N\to \infty$, the Gibbs distribution $\mathbb{P}_N^{\beta,\mu}$ converges weakly to a limiting
probability distribution $\mathbb{P}^{\beta,\mu}$.
\item[$(b)$] {\it Potts model coupled to dual CTs.} If $(\beta^*,\mu^*)\in {\bf \Sigma^*}$ then 
we have the same conclusion for the the Gibbs distribution  $\mathbb{P}^{\beta^*,\mu^*}_N$, i.e. there is 
no limiting probability measure $\mathbb{P}^{\beta^*,\mu^*}$ as $N\to\infty$.
Furthermore, if $(\beta^*,\mu^*)$ satisfied 
\begin{equation}
\mu^* > \displaystyle\frac{3}{2}\beta^* + \ln 2
+\displaystyle\frac{3}{2}\ln\left(1 + \displaystyle\frac{q^{2/3} -1}{e^{\beta^*}} \right), 
\end{equation} 
the infinite-volume free energy exists and, as $N\to \infty$, the Gibbs distribution $\mathbb{P}_N^{\beta^*,\mu^*}$ 
converges weakly to a limiting probability distribution $\mathbb{P}^{\beta^*,\mu^*}$.
\end{enumerate}
\end{theo}
As a byproduct, the Theorem \ref{theo_bounds_main2} serves 
to find lower and upper bounds  for the infinite-volume free energy. Moreover, in the case of 
$2$-state Potts model (Ising model), Theorem \ref{theo_bounds_main2} extends earlier results from
\cite{cerda}, \cite{HeAnYuZo:2013} and improves the numerical approximation of the curve in high temperature given 
in \cite{Ambjorn:1999gi}. In aditional, this approach allows to get a better aproximation of the critical curve and 
check the asymptotic behavior of the critical curve given  in \cite{Ambjorn:1999gi}, and it say that 
critical curve is  asymptotic to $\frac{3}{2}\beta + \ln2$. Furthermore, we show that the behavior for 
the matter of the free energy density in the gravitational ensemble in the thermodynamic limit, suppose in \cite{Ambjorn:1999gi} for 
the numerical simulations, is true. 
In Theorem \ref{theo_bounds_main2}, we 
find a lower and an upper curve that converges fast to $\frac{3}{2}\beta + \ln2$.

\section{FK-Potts  model on causal triangulations}\label{FK_Potts model} 
In this section we describe the FK-Potts model on causal triangulations, and in the last subsection we compute 
inequalities which will be used in the proof of the Theorem \ref{theo_duality_main1}.

\subsection{Definition of the FK-Potts model}
Now,  we turn to the FK representation of the $q$-state Potts model.
The random cluster model was originally introduced by Fortuin and Kasteleyn \cite{FK:1972} and it can be 
understood as an alternative representation of  the $q$-state Potts model. 
This representation will be referred to as the FK representation or FK-Potts model. 
We are interested in study FK-Potts model on CTs and dual CTs, and find a duality relation relation between the 
parameters of  the model on CTs and its dual in order to obtain information about the critical curve. 
In \cite{HeAnYuZo:2013}, \cite{cerda}, the model was defined putting spins on any triangle (faces), but it is  
equivalent to put spins on any vertex of the dual graph. In this section we work 
with causal triangulations and  its dual with
periodic boundary conditions, i.e., causal triangulations embedded in a torus $\mathbb{T}$ (see Figure \ref{fig1} (b)).
In general, let $G=(V,E)$ be a graph embedded in $\mathbb{T}$, we obtain its dual graph $G^*=(V^*,E^*)$  as follows: we place 
a dual vertex within each face of $G$. For each $e\in E$ we place a dual $e^*=\langle x^*,y^* \rangle$ joining the two dual 
vertices lying in the two faces of $G$ abutting $e$. Thus, $V^*$ is in one-one correspondence with  the set of faces of $G$, and
$E^*$ is a one-one correspondence with $E$. For each causal triangulation $\bft$, we denote by $\bft^*$ its dual.
\begin{figure}[t!]
\begin{center}
\includegraphics[width=8cm]{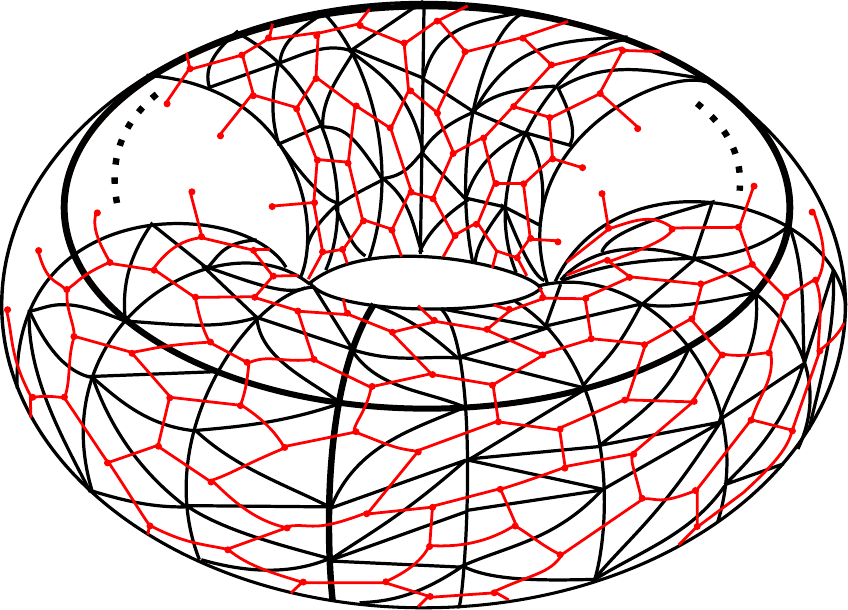}
\end{center}
\caption{ Geometric representation of a dual causal triangulation $\bft^*$
with periodic spatial boundary condition.}
\label{fig2}
\end{figure}
Let $\bft=(V(\bft),E(\bft))$ be a causal triangulation with periodic  boundary condition, where $V(\bft), E(\bft)$
denote the set of vertices and edges, respectively.
The state space for the FK-Potts model is the set   $\Sigma(\bft)=\{0,1\}^{E(\bft)}$, containing configurations that 
allocate $0's$ and $1's$  to the edge 
$e=\{i,j\}\in E(\bft)$. For $w\in\Sigma(\bft)$, we call an edge $e$ open  if $w(e)=1$, and closed  if $w(e)=0$. 
For $w\in \Sigma(\bft)$, let $\eta(w)=\{e\in E(\bft) : w(e)=1\}$ denote the set of open edges. Thus, each 
$w\in\Sigma(\bft)$  splits $V(\bft)$  into the disjoint union of maximal connected 
components, which are called the open clusters of  $\Sigma(\bft)$.
We denote by $k(w)$ the number of connected components (open clusters) of the graph $(V(\bft), \eta(w))$, and 
note that  $k(w)$ includes a count 
of isolated vertices. Two  sites of $\bft$ are said to be connected if one can be reached from another via a chain
of open bonds. 
The partition function of the FK-Potts model on $\bft$ with parameters $p$ and $q$ and periodic boundary condition  is 
defined by 
\begin{equation}
Z_{FK}(p,q,\bft)= \sum_{w\in\Sigma(\bft)}\left( \prod_{e\in E(\bft)}
(1-p)^{1-w(e)} p^{w(e)}\right)q^{k(w)},
\end{equation}
Thus, the FK-Potts measure on $\bft$ is define as follows
\begin{equation}
\Phi^{\bft}_{p,q}(w)=\displaystyle\frac{1}{Z_{FK}(p,q,\bft)}\left( \prod_{e\in E(\bft)}
(1-p)^{1-w(e)} p^{w(e)}\right)q^{k(w)}.
\end{equation}
We will use a similarly notation for the FK-Potts model on dual triangulation $\bft^*$. We 
denote by $Z_{FK}(p^*,q,\bft^*)$ and  $\Phi^{\bft^*}_{p^*,q}$ the partition function and 
the FK-Potts measure on $\bft^*$ with parameters $p^*$ and $q$, respectively.

\subsection{Edwards-Sokal coupling}

There are several ways to make the connection between the Potts  and  FK-Potts model. The correspondence
between the $q$-state Potts model and FK-Potts model was established by Fortuin and Kasteleyn \cite{FK:1972}
(see also \cite{Edwars_Sokal}, \cite{Grimmett:2006}).
In a modern approach, these two models are related via a coupling, i.e., coupled the two systems on 
a common probability space.  This coupling was introduced by Edwards-Sokal in \cite{Edwars_Sokal}.

Let $\bft$ be a CT on the cylinder $C_N$ with periodic boundary condition.
We consider the product sample space $\Omega(\bft)\times \Sigma(\bft)$ where 
$\Omega(\bft)=\{1,2,\dots,q\}^{V(\bft)}$  and $\Sigma(\bft)=\{0,1\}^{E(\bft)}$. 
The Edwards-Sokal measure $\mathcal{Q}$ on $\Omega(\bft)\times \Sigma(\bft)$ is define by 
$$\mathcal{Q}(\sigma,w) \propto \prod_{e=\{i,j\}\in E(\bft)} \left\{ (1-p)\delta_{w(e),0}  
+ p\delta_{w(e),1}\delta_{\sigma_i,\sigma_j} \right\}$$

\begin{theo}[Edwards-Sokal \cite{Edwars_Sokal}]\label{Edwars_Sokal_coupling}
Let $q\in\{2,3,\dots\}$.  Let $p\in (0,1)$ and $\bft$ a CT with periodic boundary condition, and 
suppose that $p=1-e^{-\beta}$. 
If the configuration $w$ is distributed according to an FK-Potts measure with parameters $(p,q)$ on $\bft$, then
$\bfsigma$ is distributed according to a  $q$-state Potts measure with inverse temperature $\beta$. Furthermore, 
the Edwards-Sokal measure provides a coupling  of $\mu^{\bft}_{\beta,q}$ and $\Phi^{\bft}_{p,q}$, i.e.
$$\sum_{w\in\Sigma(\bft)} \mathcal{Q}(\sigma,w) = \mu^{\bft}_{\beta,q}(\sigma),$$
for all $\sigma\in\Omega(\bft)$, and 
$$\sum_{\sigma\in\Omega(\bft)} \mathcal{Q}(\sigma,w) = \Phi^{\bft}_{p,q}(w),$$
for all $w\in\Sigma(\bft)$. Moreover, we have the relation between partition functions
\begin{equation}\label{relat_partition_functions}
Z_{FK}(p,q,\bft) = e^{-\beta |E(\bft)|} Z_P (\beta,q,\bft).
\end{equation}
\end{theo}

\subsection{FK-Potts model coupled to CTs with periodic boundary conditions}

In this section we obtain a relation between the partition functions of FK-Potts model on a  triangulation
$\bft$ and its dual. This relation was studied by Beffara and Duminil-Copin for the FK-Potts model on  $\Z^2$ with free, wired and 
periodic boundary condition (see \cite{Beffara_Duminil}). We will view wich the dual of a FK-Potts model defined on a torus is an
almost FK-Potts model, but it is not very different from one. 

Let $\bft$ and $\bft^*$ a CT with periodic boundary condition and its dual.  Each configuration 
$w\in\Sigma(\bft)=\{0,1\}^{E(\bft)}$ gives rise to a dual configuration $w^*\in\Sigma(\bft^*)=\{0,1\}^{E(\bft^*)}$
given by $w^*(e^*)=1-w(e)$. That is, $e^*$ is declared open if and only if the corresponding bond  $e$ is closed.  
The new configuration $w^*$ is called the dual configuration of $w$, and note that there exists an one-one correspondence 
between $\Sigma(\bft)$ and $\Sigma(\bft^*)$.
As in the Section \ref{FK_Potts model},
to each configuration $w^*$ there corresponds  the set $\eta(w^*)=\{e^*\in E(\bft^*) : w^*(e^*)=1\}$ of its open edges.

Now, we start of one FK-Potts model on $\bft$ with parameters $p$ and $q$, and  we will obtain two FK-Potts models
on the dual triangulation $\bft^*$ with parameters $p^*$ and $q$ that limited the initial model.

Let  $o(w)$ (resp. $c(w)$) denote the number of open  edges (resp. closed) of $w$, $k(w)$ the number of connected components
of $w$, and  $f(w)$ the number of faces delimited by $w$, i.e. the number  of connected components of the complement 
of the set of open bonds. Now, We will define an additional parameters $\delta(w)$ which is associated with the topology 
of the surface where the graph was embedded.
Call a connected component of $w$ a {\it net} if it contains two non-contractible simple loops $\gamma_1, \gamma_2$ of 
different homotopy  classes, 
and a {\it cycle}  if it contain a non-contractible simple loops $\gamma_1$  but is not 
a net. These definitions were introduced in \cite{Beffara_Duminil}. In aditional,
notice that every configuration $w$ can be of one three types: 

\begin{itemize}
 \item One of the cluster of $w$ is a net. 
 In that case, we let $\delta(w)=2$;
 \item One of the cluster of $w$ is a cycle.  
 We then let $\delta(w)=1$;
 \item None of the cluster of $w$ is a net or a cycle. We let $\delta(w)=0$.
\end{itemize}
Utilizing the previous definition of the parameter $\delta$, we obtained the following version of Euler's formula.
\begin{prop}[Euler's formula]\label{Euler}
Let $\bft$ a CT with periodic boundary condition and $w\in \{0,1\}^{E(\bft)}$. Then 
\begin{equation}\label{euler_formula}
|V(\bft)|- o(w) + f(w) = k(w) + 1 - \delta(w). 
\end{equation}
\end{prop}
\noindent Employing duality and Proposition \ref{Euler}, we have the following relations 
\begin{equation}\label{consequences_Euler}
o(w)+o(w^*)=|E(\bft)|, \quad f(w)=k(w^*)\;\; \mbox{and}\;\; \delta(w)+\delta(w^*)=2. 
\end{equation}
Let $q\in [2,\infty)$ and $p\in (0,1)$. The partition function for the  FK-Potts model is given by 
$$
\begin{array}{ccl}
Z_{FK}(p,q,\bft) &=& \displaystyle\sum_{w\in\Sigma(\bft)}\left( \prod_{e\in E(\bft)}
(1-p)^{1-w(e)} p^{w(e)}\right)q^{k(w)}\\
  &=& \displaystyle\sum_{w\in\Sigma(\bft)} p^{o(w)}(1-p)^{c(w)}  q^{k(w)}.
\end{array}$$
Employing Euler's formula and relations (\ref{consequences_Euler}), we  rewrite the number of cluster of $w$ in terms of its dual $w^*$
$$k(w)= |V(\bft)|- |E(\bft)| + o(w^*) + k(w^*) + 1 - \delta(w^*).$$
Note also that  $o(w)+o(w^*)=|E(\bft)|=|E(\bft^*)|$. Plugging before relations into the partition function of the  FK-Potts 
model, we obtain
$$
\begin{array}{ccl}
Z_{FK}(p,q,\bft) &=&  p^{|E(\bft)|}q^{|V(\bft)|- |E(\bft)|} \displaystyle\sum_{w\in\Sigma(\bft)}\left( \frac{1-p}{p} \right)^{o(w^*)} q^{o(w^*) + k(w^*) + 1 - \delta(w^*)}\\  
\end{array}$$
As there exists an one-one correspondence between $\Sigma(\bft)$ and $\Sigma(\bft^*)$, in the last equality we 
we change the sum over  $\Sigma(\bft)$ by the   sum  over $\Sigma(\bft^*)$. Thus, we obtain the following representation of the 
partition function in terms of configurations into $\Sigma(\bft^*)$
\begin{equation}\label{fkk}
Z_{FK}(p,q,\bft) = p^{|E(\bft)|}q^{|V(\bft)|- |E(\bft)|} \!\!\!\displaystyle\sum_{w^*\in\Sigma(\bft^*)}
\left( \frac{q(1-p)}{p} \right)^{o(w^*)}\!\! q^{k(w^*) + 1 - \delta(w^*)}
\end{equation}
Using the  relation (\ref{fkk}),   we obtain the following lemma.
\begin{lm}\label{duality}
Let $\bft$ be a CT with periodic boundary condition. Then the following comparison inequalities both
\begin{equation}
Z_{FK}(p,q,\bft)\leq \left(\frac{p}{1-p^*} \right)^{|E(\bft)|}q^{|V(\bft)|- |E(\bft)| +1} Z_{FK}(p^*,q,\bft^*) 
\end{equation}
and 
\begin{equation}
 \left(\frac{p}{1-p^*} \right)^{|E(\bft)|}q^{|V(\bft)|-|E(\bft)|-1} Z_{FK}(p^*,q,\bft^*)  \leq Z_{FK}(p,q,\bft)
\end{equation}
where  $Z_{FK}(p^*,q,\bft^*)$ is the partition function for FK-Potts model on $\bft^*$ with parameters $q$ and 
$p^*=p^*(p,q)$ satisfying
$$p^*(p,q)=\frac{(1-p)q}{(1-p)q + p},\;\;\mbox{or equivalently}\;\; \frac{p^*}{1-p^*}\cdot \frac{p}{1-p} =q.$$
\end{lm}
\begin{proof}
We  introduce the parameter $p^*=p^*(p,q)$ as solution of the equation 
$$\displaystyle\frac{p^*}{1-p^*} =\frac{(1-p)q}{p}.$$ 
and it is plugging in equation (\ref{fkk}). Thus, the partition function can be written in the following ways
$$
\begin{array}{ccl}
Z_{FK}(p,q,\bft) &=& \displaystyle\frac{p^{|E(\bft)|}}{(1-p^*)^{|E(\bft^*)|}}q^{|V(\bft)|- |E(\bft)|} 
      \displaystyle\sum_{w^*\in\Sigma(\bft^*)} (p^*)^{o(w^*)}(1-p^*)^{c(w^*)} q^{k(w^*) + 1 - \delta(w^*)}.
\end{array}$$
Notice that $-1\leq 1-\delta(w^*)\leq 1$, for any $w^*\in\Sigma(\bft^*)$. We denote by $Z_{FK}(p^*,q,\bft^*)$, the partition 
function of a FK-Potts model with parameters $p^*$ and $q$. Thus, we obtain the upper bound 
$$
\begin{array}{ccl}
Z_{FK}(p,q,\bft) &\leq&  \displaystyle\frac{p^{|E(\bft)|}}{(1-p^*)^{|E(\bft^*)|}}q^{|V(\bft)|- |E(\bft)| +1} Z_{FK}(p^*,q,\bft^*)
\end{array},$$
and the lower bound
$$
\begin{array}{ccl}
\displaystyle\frac{p^{|E(\bft)|}}{(1-p^*)^{|E(\bft^*)|}}q^{|V(\bft)|- |E(\bft)| -1} Z_{FK}(p^*,q,\bft^*) &\leq& Z_{FK}(p,q,\bft) 
\end{array}$$
for the partition function of FK-Potts model on $\bft$ with parameters $p$ and $q$. 
Using the one-one correspondence between $E(\bft)$ and $E(\bft^*)$, we conclude the proof.
\end{proof}

The partition function for pure CT's has been determined as a sum over all possible triangulations of a cylinder where
each configuration is weighted by a Boltzmann factor $e^{-\mu n(\bft)}$, where $n(\bft)$ standing for the size of the
triangulation and $\mu$ being the cosmological constant. Thus, into two-dimensional quantum gravity the volume $n(\bft)$ 
becomes an important dynamical variable for the model. Therefore, we rewrite inequalities for the 
partition function in Lemma \ref{duality} in terms of the dynamical variable $n(\bft)$.  In the Table \ref{table1} 
we show the simple relation among $V(\bft), E(\bft), V(\bft^*), E(\bft^*)$ and  the number of triangles $n(\bft)$ of  
a CT $\bft$.
\begin{table}[!htpb]
\centering
\begin{footnotesize}
\setlength{\tabcolsep}{15pt} 
\begin{tabular}{|c|c|}
\hline 
$\quad\bft=(V(\bft),E(\bft))\quad$ & $\quad\bft^*=(V(\bft^*),E(\bft^*))\quad$ \\
\hline \hline
$|V(\bft)| = \displaystyle\frac{1}{2}n(\bft)$ &  $|V(\bft^*)| =  n(\bft)$ \\ [1ex]
\hline
$|E(\bft)| = \displaystyle\frac{3}{2}n(\bft)$ & $|E(\bft^*)| = \displaystyle\frac{3}{2}n(\bft)$ \\ [1ex]
\hline
$|\mbox{faces in}\; \bft|=n(\bft)$ &  $|\mbox{faces in}\; \bft^*|=\displaystyle\frac{1}{2}n(\bft)$  \\ [1ex]
\hline
\end{tabular}
\end{footnotesize}
\caption{Relation between the graphs $\bft, \bft^*$ and   $n(\bft)$}
\label{table1}
\end{table}
Employing relations of Table \ref{table1}, the Lemma \ref{duality}  becomes be written  in terms of $n(\bft)$ as follow
\begin{cor}
Let $\bft$ be a CT with periodic boundary condition. Then the
following comparison inequalities both 
\begin{equation}\label{eq2.25}
\left(\frac{p}{1-p^*} \right)^{\frac{3}{2}n(\bft)}q^{-1-n(\bft)}\leq   \frac{Z_{FK}(p,q,\bft)}{ Z_{FK}(p^*,q,\bft^*)}   \leq \left(\frac{p}{1-p^*} \right)^{\frac{3}{2}n(\bft)}q^{1-n(\bft)} 
\end{equation}
and 
\begin{equation}\label{eq2.26}
\left(\frac{p^*}{1-p} \right)^{\frac{3}{2}n(\bft)}q^{-1-\frac{1}{2}n(\bft)}\leq   \frac{Z_{FK}(p^*,q,\bft^*)}{Z_{FK}(p,q,\bft)}   \leq \left(\frac{p^*}{1-p} \right)^{\frac{3}{2}n(\bft)}q^{1-\frac{1}{2}n(\bft)} 
\end{equation}
where parameters $q$ and 
$p^*=p^*(p,q)$ satisfy
$$p^*(p,q)=\frac{(1-p)q}{(1-p)q + p},\;\;\mbox{or equivalently}\;\; \frac{p^*}{1-p^*}\frac{p}{1-p} =q.$$
\end{cor}

\section{Proof of Theorem \ref{theo_duality_main1}}\label{Sect3} 

In the previous section we have found comparison inequalities between the partition function of the FK-Potts model 
on $\bft$ and the partition function of the FK-Potts model on its dual. In this section, we employ
results of previous sections in order to prove Theorem \ref{theo_duality_main1}.  
Combining inequalities (\ref{eq2.25}), (\ref{eq2.26}) and the Edwars-Sokal coupling, 
we obtain the following comparison inequalities between the partition function of the $q$-state Potts model 
on $\bft$ and the partition function of the $q$-state Potts model on its dual $\bft^*$.

\begin{equation}\label{eq2.27}
\left(\frac{p}{1-p^*} \right)^{\frac{3}{2}n(\bft)}q^{-1-n(\bft)} e^{\frac{3}{2}(\beta-\beta^*)n(\bft)} \leq   
\frac{Z_{P}(\beta,q,\bft)}{ Z_{P}(\beta^*,q,\bft^*)}   
\leq \left(\frac{p}{1-p^*} \right)^{\frac{3}{2}n(\bft)}q^{1-n(\bft)} e^{\frac{3}{2}(\beta-\beta^*)n(\bft)}
\end{equation}
and 
\begin{equation}\label{eq2.28}
\left(\frac{p^*}{1-p} \right)^{\frac{3}{2}n(\bft)}q^{-1-\frac{1}{2}n(\bft)} e^{\frac{3}{2}(\beta^*-\beta)n(\bft)} \leq   
\frac{Z_{P}(\beta^*,q,\bft^*)}{Z_{P}(\beta,q,\bft)}   \leq 
\left(\frac{p^*}{1-p} \right)^{\frac{3}{2}n(\bft)}q^{1-\frac{1}{2}n(\bft)}e^{\frac{3}{2}(\beta^*-\beta)n(\bft)} 
\end{equation}
where   $(e^\beta - 1 )(e^{\beta^*}-1)=q$. 

Note that inequalities (\ref{eq2.27}) and (\ref{eq2.28}) are satisfied for any $\bft\in \LT_N$  and for any $N\in\N$. Utilizing this
observation we establish the following relation:

Let $\bft$, $\bft^*$ be a infinite causal triangulation and its dual respectively, and denoting by $\psi_G(\beta)$ the free energy
of the Potts model defined on the graph $G=\bft, \bft^*$. Then, we obtain the following result in the thermodynamic limit.
\begin{equation}
\begin{array}{ccl}
\displaystyle\frac{1}{2}\psi_{\bft}(\beta) - \psi_{\bft^*}(\beta^*) &=& \displaystyle\frac{3}{2}\ln\left( 1-e^{-\beta}+qe^{-\beta} \right) - \ln q + 
                 \frac{3}{2}(\beta-\beta^*)\\
                 &=& \displaystyle\frac{3}{2}\ln\left(e^{\beta} -1 \right) - \ln q,
\end{array}
\end{equation}
for any $\bft\in\LT_{\infty}$.

\begin{proof}[Proof of Theorem \ref{theo_duality_main1}.] 
Remember that $p^*=1-e^{-\beta^*}$ and $p=1-e^{-\beta}$. Thus, 
$$\displaystyle\frac{p^*}{1-p}= (1-e^{-\beta^*}) + q e^{-\beta^*}= \frac{q}{(1-e^{-\beta}) + q e^{-\beta}}$$
and 
$$\displaystyle\frac{p}{1-p^*}= \frac{q}{(1-e^{-\beta^*}) + q e^{-\beta^*}}=  (1-e^{-\beta}) + q e^{-\beta}.$$
Multiplying by the Boltzmann factor $e^{-\mu n(\bft)}$ in equations (\ref{eq2.27}) and (\ref{eq2.28}), and  sum over all possible 
CTs of the  cylinder  $C_N$, we obtain the following comparison inequalities between the annealed model 
\begin{equation}\label{eq2.29}
\frac{1}{q} \Xi_N^*(\beta^*,\mu^*) \leq \Xi_N(\beta,\mu) \leq q \Xi_N^*(\beta^*,\mu^*)
\end{equation}
where $\Xi_N$ and $\Xi_N^*$ were defined in (\ref{yamb-pf}) and (\ref{yamb-pf_dual}), and they denote the 
partition function of the $q$-state Potts  model coupled to dual CTs  and its dual, respectively, with parameters
related by 
$$\beta^*= \ln\left( 1+ \frac{q}{e^\beta -1} \right), \quad \mu^* = \mu - \frac{3}{2}\ln(e^\beta -1) + \ln q.$$
Similarly, we have
\begin{equation}\label{eq2.30}
\frac{1}{q} \Xi_N(\beta,\mu) \leq \Xi_N^*( \beta^*,\mu^*) \leq q \Xi_N(\beta,\mu)
\end{equation}
where 
$$\beta= \ln\left( 1+ \frac{q}{e^{\beta^*} -1} \right), \quad \mu= \mu^* - \frac{3}{2}\ln(e^{\beta^*} -1) + \frac{1}{2}\ln q.$$
Taking the natural  logarithm in inequalities (\ref{eq2.29}) and (\ref{eq2.30}), divide both sides of the above inequalities by $N$,
we obtain the inequality
$$\left| \frac{\ln\Xi_N(\beta,\mu) -\ln\Xi_N^*( \beta^*,\mu^*)  }{N} \right| \leq \frac{\ln q}{N}.$$
Let $N\to \infty$ we conclude the proof of Theorem \ref{theo_duality_main1}. 
\end{proof}
An interesting and simple consequence of the  Theorem \ref{theo_duality_main1} is the
asymptotic behavior of the parameter associated with the random geometry in the coupled model.
\begin{cor}\label{asympt_beh}
The following asymptotic behavior for parameters $\mu$ and $\mu^*$ are fulfilled.
\begin{enumerate} 
\item If small $\beta$ and $\beta^*$, we have that 
$$\mu(\beta)\approx \ln2\sqrt{q} +o(\beta^2) \quad, \quad \mu(\beta)\approx \ln2q +o(\beta^2) .$$
\item If large $\beta$ and $\beta^*$, we have that 
$$\mu(\beta)\approx \frac{3}{2}\beta + \ln2\quad, \quad \mu^*(\beta^*)\approx \frac{3}{2}\beta^* + \ln2.$$
\end{enumerate}
\end{cor}

\section{Proof of Theorem \ref{theo_bounds_main2}}\label{proof_theo2}

\subsection{Lower bound for the critical curve}
Let $\bft$ be a CT embedded in the torus with height $N$.  We define the set $\Pi_i$
of configurations in $\Sigma(\bft)$ which splits $V(\bft)$ in $i$ maximal connected components, i.e.
$$\Pi_i= \{w\in\Sigma(\bft) :  k(w)=i\}.$$
Similarly, we denote $\Pi_i^*$ the set of configurations in $\Sigma(\bft^*)$ which splits 
$V(\bft^*)$ in $i$ maximal connected components.

Another way of writing $Z_P$ is as the moment generating function of the number
of clusters in a random graph as follow
\begin{equation}\label{Z_1}
Z_P (\beta,q,G)= \exp(\beta |E(G)|) \phi^{G}_{p}\left( q^{k(w)}\right) =\exp(\beta |E(G)|)
\displaystyle\sum_{i\geq 1} q^i \phi^{G}_{p}\left(\Pi_i\right),
\end{equation}
where $|E(G)|$ denotes the number of edges of the graph $G$, and $\phi_p^{G}$ 
denotes product measures on $\Sigma(G)=\{0,1\}^{E(G)}$.
 
Utilizing $G=\bft, \bft^*$, we write the representation (\ref{Z_1}) in terms of the 
dynamical variable $n(\bft)$ of the model. For that, we consider the two cases of interest separately.
\vspace{0.5cm}

{1. \it The model on CTs}: In this case, we can to write the partition function in 
terms of volume $n(\bft)$ of the triangulation as follow
\begin{equation}\label{Z_1_1}
Z_P (\beta,q,\bft) = \exp\left(\displaystyle\frac{3}{2} \beta n(\bft)\right) \phi^{\bft}_{p}\left( q^{k(w)}\right).
\end{equation}
Utilizing representation (\ref{Z_1_1}), we obtain two lower bounds for the partition function  of the Potts model  on $\bft$
\begin{equation}\label{lower}
Z_P (\beta,q,\bft) \geq \max\left\{ q \left(e^\beta -1\right)^{\frac{3}{2}n(\bft)},
 q^{\frac{1}{2}n(\bft)} \right\}.
\end{equation}
These lower bounds for the $q$-state Potts model on $\bft$ permit to obtain a lower barrier for the parameters where 
the annealed model could be defined, and the partition function of the model coupled to CTs could no 
explode in finite volume. 
Apparently, the reader could think that crude bounds employed do not give thermodynamics information of the model, but these
bound are sharp at high and low temperature, as we show in the proof of Theorem \ref{theo_bounds_main2}. Furthermore,
these  lower bounds serves in order to obtain information about the Gibbs measure for $q$-state Potts model coupled to CTs.
\begin{prop}\label{lowerbound}
If $(\beta,\mu)\in\R^+$ such that 
$$\mu<\max\left\{ \ln 2\sqrt{q}, \displaystyle\frac{3}{2}\ln (e^\beta -1)+\ln2 \right\},$$
then there exists $N_0\in\N$ such that the partition function $\Xi_N(\beta,\mu)=\infty$ whenever
$N>N_0$. Moreover, the Gibbs distribution $\mathbb{P}_N^{\beta,\mu}$ cannot be defined
by using the standard formula with $\Xi_N(\beta,\mu)$ as a normalising denominator, consequently, 
there is no limiting probability measure $\mathbb{P}^{\beta,\mu}$ as $N\to\infty$.
\end{prop}
\begin{proof}
The lower bound in (\ref{lower}) for $Z_P (\beta,q,\bft)$ provide the following lower bounds for the annealed partition
function $\Xi_N (\beta,\mu)$, 
\begin{equation}\label{lower_anneal}
\Xi_N (\beta,\mu) \geq \max\left\{q Z_N\left(\mu - \frac{3}{2}\ln(e^\beta -1)\right), Z_N\left(\mu - \frac{1}{2}\ln q\right)\right\}.
\end{equation}
Employing the estimation  (\ref{yamb-e14}), we obtain which the partition function $\Xi_N(\beta,\mu)$ there is no
exist if 
$$\mu \leq \frac{1}{2}\ln q +\ln\left(2\cos\frac{\pi}{N+1}\right) \; \mbox{or}\; 
\mu\leq \frac{3}{2}\ln (e^\beta -1)+\ln\left(2\cos\frac{\pi}{N+1}\right).$$
Letting $N\to \infty$ we conclude the proof.
\end{proof}

Proposition \ref{lowerbound} provide a region where the model cannot be defined, and the partition function of the 
$q$-state Potts model coupled to CTs is infinite.
Denote by ${\bf \Sigma}$ points
$(\beta,\mu)$ in $\R_+^2$ with the condition of Proposition \ref{lowerbound}. Then, by the duality relation established
in Theorem \ref{theo_duality_main1} we obtain a region ${\bf \Sigma^*}$ where the model coupled to dual CTs 
also cannot be defined (see Table \ref{table2}). In the next section we will prove that this lower bound is sharp 
at low and high temperature.

\begin{table}[!htpb]
\centering
\begin{footnotesize}
\setlength{\tabcolsep}{15pt} 
\begin{tabular}{|c|c|c|}
\hline 
\mbox{from CTs}\;\; $\bft$  & $\xrightarrow{\mbox{by duality}}$  & \mbox{to dual CTs} \;\; $\bft^*$ \\
\hline \hline
$\mu<\displaystyle\frac{1}{2}\ln q +\ln2$ & $\rightarrow$  & $\mu^* < \displaystyle\frac{3}{2}\ln (e^{\beta^*} -1)+\ln2$ \\ [1ex]
\hline
$\mu < \displaystyle\frac{3}{2}\ln (e^\beta -1)+\ln2$ & $\rightarrow$  &  $\mu^* < \ln q +\ln2$ \\ [1ex]
\hline
\end{tabular}
\end{footnotesize}
\caption{Bounds for the critical curve of the $q$-state Potts model on CTs generated bounds on its dual.}
\label{table2}
\end{table}

{2. \it The model on dual CTs}: Similarly, we  write the partition function in 
terms of volume $n(\bft)$ of the triangulation as follow
\begin{equation}\label{Z_2_2}
Z_P (\beta^*,q,\bft^*) = \exp\left(\displaystyle\frac{3}{2} \beta^* n(\bft)\right) \phi^{\bft^*}_{p}\left( q^{k(w)}\right).
\end{equation}
and utilizing  this representation, we obtain the following lower bounds for the partition function  of the Potts model  on $\bft^*$
\begin{equation}
Z_P (\beta,q,\bft) \geq \max\left\{ q \left(e^{\beta^*} -1\right)^{\frac{3}{2}n(\bft)},
 q^{n(\bft)} \right\}.
\end{equation}
In similar way as in before case, we have the following assertion about non existence of Gibbs 
measures for the model.
\begin{prop}\label{lower_dual} 
If $(\beta^*,\mu^*)\in\R^+$ such that 
$$\mu^*<\max\left\{ \ln 2q,\frac{3}{2}\ln (e^{\beta^*} -1)+\ln2 \right\},$$
then there exists $N_0\in\N$ such that the partition function $\Xi_N^*(\beta^*,\mu^*)=\infty$ whenever
$N>N_0$. Moreover, the Gibbs distribution $\mathbb{P}_N^{\beta^*,\mu^*}$ cannot be defined
by using the standard formula with $\Xi_N^*(\beta^*,\mu^*)$ as a normalising denominator, consequently, 
there is no limiting probability measure $\mathbb{P}^{\beta^*,\mu^*}$ as $N\to\infty$.
\end{prop}

Utilizing  the duality relation of Theorem \ref{theo_duality_main1}, Proposition \ref{lower_dual} provide 
a region where the Potts model coupled to CTs cannot be defined (see Table \ref{table3}).
\begin{table}[!htpb]
\centering
\begin{footnotesize}
\setlength{\tabcolsep}{15pt} 
\begin{tabular}{|c|c|c|}
\hline 
\mbox{from dual CTs}\;\; $\bft^*$  & $\xrightarrow{\mbox{by duality}}$  & \mbox{to CTs} \;\; $\bft$ \\
\hline \hline
$\mu^* < \ln q +\ln2$ & $\rightarrow$  & $\mu < \displaystyle\frac{3}{2}\ln \left(e^{\beta} -1\right)+\ln2$ \\ [1ex]
\hline
$\mu^* < \displaystyle\frac{3}{2}\ln \left(e^{\beta^*} -1\right)+\ln2$ & $\rightarrow$  &  $\mu < \displaystyle\frac{1}{2}\ln q +\ln2$ \\ [1ex]
\hline
\end{tabular}
\end{footnotesize}
\caption{Bounds for the critical curve of the $q$-state Potts model coupled to dual CTs generated bounds  for the the critical 
curve of the $q$-state Potts model coupled to CTs.}
\label{table3}
\end{table}

Introducing the function $\varphi: \R_+^2 \to \R_+^2$, given by 
\begin{equation}
\varphi(\beta,\mu)=\left(\ln\left(1+\frac{q}{e^\beta -1}\right), 
\mu-\frac{3}{2}\ln\left(e^\beta -1\right)+ \ln q\right),
\end{equation}
and utilizing sets ${\bf \Sigma}$ and ${\bf \Sigma^*}$, defined in Theorem \ref{theo_duality_main1}, we have 
that $\varphi({\bf \Sigma})={\bf \Sigma^*}$.
Employing duality relation, in the next section we compute an upper bound for the critical curve of the Potts model on CTs 
and its dual. In aditional, this approach allows to get behavior of the critical curve for the annealed model for low and 
high temperature.

\subsection{Upper bound for the critical curve}\label{Sect4}
In this section we utilize a High-T expansion for the Potts model introduced by Domb in \cite{Domb}, see also the review
\cite{Baxter}, \cite{Wu}.
 
Let $\bft$ be  a CT with periodic boundary condition.  The partition  function for the Potts 
model on $\bft$ is write in the usual high-T expansion as
\begin{equation}\label{eq2.32}
Z_P(\beta,q,\bft)= \left( \frac{q+h}{q}\right)^{|E(\bft)|} \sum_{\sigma} \prod_{\langle i,j\rangle}(1+ f_{ij})
\end{equation}
where $h=e^\beta -1$ and $f_{ij}=\frac{h}{q+h} (-1 + q\delta_{\sigma_i,\sigma_j})$. It can be readily verified that 
$\sum_{\sigma} f_{ij}=0$ for all $\{i,j\}\in E(\bft)$, consequently, all subgraphs with one or more vertices of degree $1$ 
give rise to zero contributions. Thus,  the partition function can be  written  as follow
$$
\begin{array}{ccl}
Z_P(\beta,q,\bft) &=& \left( \displaystyle\frac{q+h}{q}\right)^{|E(\bft)|} \displaystyle\sum_{\sigma} 
\displaystyle\sum_{A\in G(\bft)} \prod_{\{ i,j\}\in A}f_{ij},
\end{array}$$
where $G(\bft)$ is the set of families of edges of $\bft$ without vertices of degree $1$. Therefore, 
we can rewrite the partition function as 
$$
\begin{array}{ccl}
Z_P(\beta,q,\bft) &=& \left( \displaystyle\frac{q+h}{q}\right)^{|E(\bft)|} 
\displaystyle\sum_{A\in G(\bft)} w(A)   
\end{array}$$
where  $w(A)=\displaystyle\sum_{\sigma}  \prod_{\{ i,j\}\in A}f_{ij}$ is a weight factor associated with 
the subset $A$. We then proceeded to determine $w(A)$. An expression of $w(A)$ for general $A$ can be 
obtained by further expanding in $w(A)$ the product $\displaystyle\sum_{\sigma}  \prod_{\{ i,j\}\in A}f_{ij}$. This 
procedure leads to
$$w(A)=\left( \frac{h}{q+h} \right)^{|A|} \displaystyle\sum_{\sigma}  \mathcal{P}(A)(\sigma),$$
where $\mathcal{P}(A)(\sigma)= \prod_{e\in A} (-1 + q\delta_e(\sigma))$, and if   $e=\{ i,j\}$ then
$\delta_e(\sigma)=\delta_{\sigma_i,\sigma_j}$. Expanding $\mathcal{P}(A)(\sigma)$, we have the following representation
$$
\begin{array}{ccl}
\mathcal{P}(A)(\sigma) &=& (-1)^{|A|}  +(-1)^{|A|-1}q\displaystyle\sum_{e\in A}\delta_e(\sigma) +
        (-1)^{|A|-2}q^2\displaystyle\sum_{e_1,e_2\in A}\delta_{e_1}(\sigma)\delta_{e_2}(\sigma)\\
     & & +\cdots+(-1)q^{|A|-1}\displaystyle\sum_{e_1,\dots,e_{|A|-1}\in A}
         \delta_{e_1}(\sigma)\dots \delta_{e_{|A|-1}}(\sigma)\\
     & &   + q^{|A|}\delta_{e_1}(\sigma)\dots \delta_{e_{|A|}}(\sigma).
\end{array}$$
\begin{figure}[t!]
\begin{center}
\includegraphics[width=10cm]{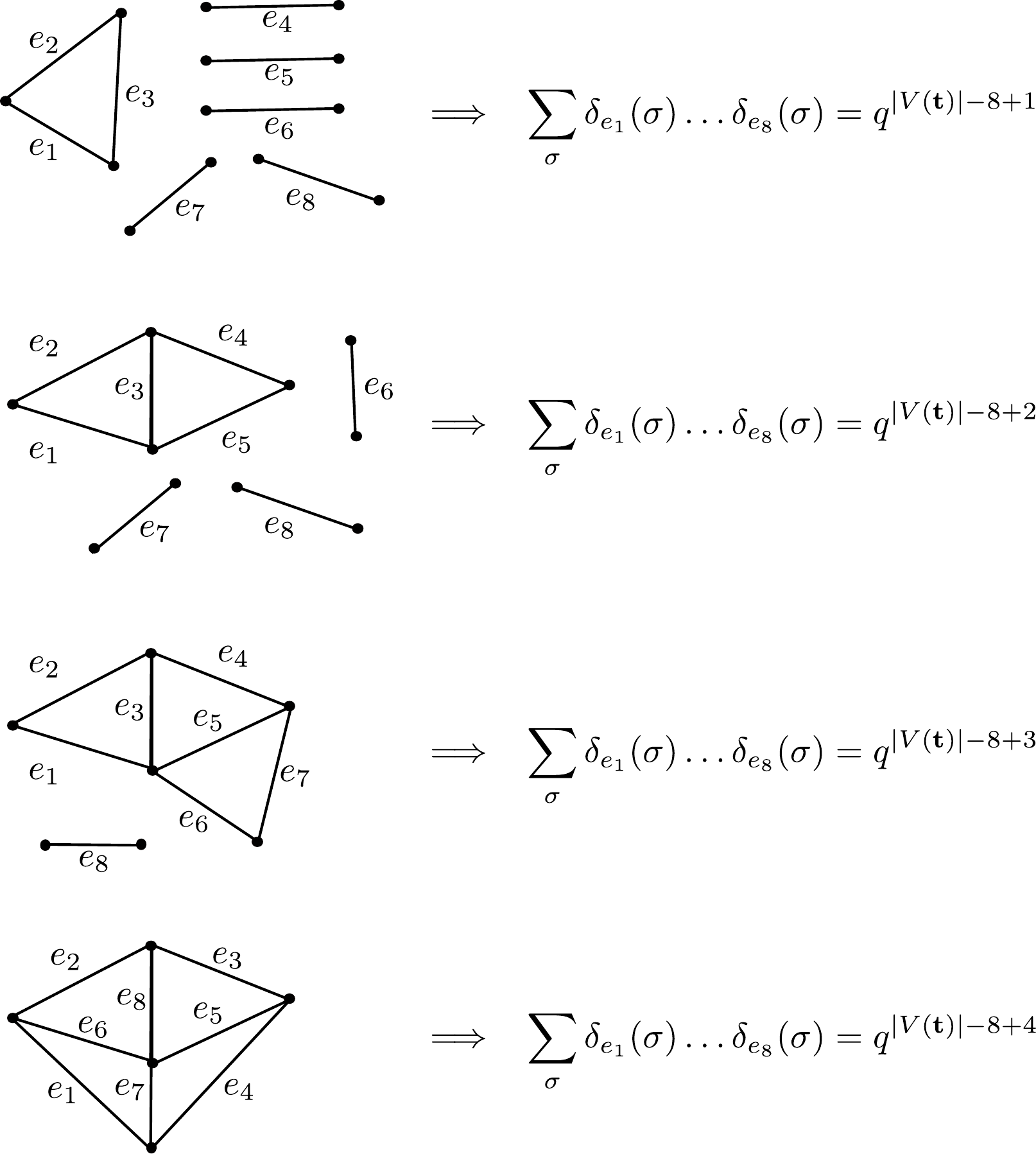}
\end{center}
\caption{Examples of three subgraphs of $A$ with 8 edges. It is clear  that the term $\xi(e_1,\dots,e_8)$ depends of 
the topology of the subgraphs.}
\label{fig4}
\end{figure}

We choose $k$ edges $\{e_1,\dots,e_k\}$ of $A$. These edges form a subgraph of $A$. Thus, we obtain 
$$\sum_{\sigma} \delta_{e_1}(\sigma)\dots\delta_{e_k}(\sigma) = q^{|V(\bft)| - k + \xi(e_1,\dots,e_k)}$$
where $\xi(e_1,\dots,e_k)$ stands the total numbers of  internal faces in each maximal connected 
component of  $\{e_1,\dots,e_k\}$ (number of independent circuits in $\{e_1,\dots,e_k\}$).
Note that 
this terms depends essentially on the topology of  $\{e_1,\dots,e_k\}$ (see Figure \ref{fig4}), but
$\xi(e_1,\dots,e_k)\leq \displaystyle\frac{2}{3}(k+1)$ for all $k$. Thus, we obtain the estimate
$$\sum_{\sigma} \delta_{e_1}(\sigma)\dots\delta_{e_k}(\sigma) \leq  q^{|V(\bft)| - k + \frac{2}{3}(k+1)}=
q^{|V(\bft)| -\frac{k}{3} + \frac{2}{3}}$$
and 
$$\sum_{\sigma} \sum_{e_1,\dots e_k \in A} \delta_{e_1}(\sigma)\dots\delta_{e_k}(\sigma) \leq \binom{|A|}{k}
q^{|V(\bft)| -\frac{k}{3} + \frac{2}{3}}.$$
Therefore, 
$$
\begin{array}{ccl}
\displaystyle\sum_{\sigma}\mathcal{P}(A)(\sigma) &\leq&  q^{|V(\bft)|+ \frac{2}{3}} \displaystyle\sum_{k=0}^{|A|} \binom{|A|}{k} 
(-1)^{|A|-k} (\sqrt[3]{q^2})^k = q^{|V(\bft)|+ \frac{2}{3}}  (\sqrt[3]{q^2} - 1)^{|A|}, 
\end{array}$$
and 
$$
\begin{array}{ccl}
Z_P(\beta,q,\bft) &\leq& \left( \displaystyle\frac{q+h}{q}\right)^{|E(\bft)|}  q^{|V(\bft)|+ \frac{2}{3}} 
\displaystyle\sum_{A\in G(\bft)} \left( (\sqrt[3]{q^2}-1) \displaystyle\frac{h}{q+h} \right)^{|A|}\\
  &\leq& \left( \displaystyle\frac{q+h}{q}\right)^{|E(\bft)|}  q^{|V(\bft)|+ \frac{2}{3}} 
  \left(1 + \displaystyle\sum_{k\geq 1} \Omega_k(\bft)u^k\right),
\end{array}$$
where  $\Omega_k(\bft)= | \{A\in G(\bft) : |A|=k\} |$ and $u=(\sqrt[3]{q^2}-1) \displaystyle\frac{h}{q+h}$. An simple estimation
establish that $\Omega_k(\bft) \leq \binom{|E(\bft)|}{k}$. Thus, we obtain the estimate 
\begin{equation}\label{eq2.33}
Z_P(\beta,q,\bft) \leq \left( \displaystyle\frac{q+h}{q}\right)^{|E(\bft)|}  q^{|V(\bft)|+ \frac{2}{3}} (1 + u)^{|E(\bft)|}.
\end{equation}
\begin{proof}[Proof of Theorem \ref{theo_bounds_main2}]
Employing the inequality (\ref{eq2.33}) and Table \ref{table1}, we write the bound (\ref{eq2.33}) for the partition
function of the Potts model on $\bft$ in terms of the number of triangles $n(\bft)$. This inequality is true for any graph,
therefore similar computations serve for the dual model. We consider the two cases of interest separately.

\vspace{0.5cm}

{1. \it The model on CTs}: In this case, we can to write the bound (\ref{eq2.33}) for the partition function in 
terms of volume $n(\bft)$ of the triangulation as follow
\begin{equation}\label{Z_N_B}
Z_P (\beta,q,\bft) \leq \left( \displaystyle\frac{q+h}{q}\right)^{\frac{3}{2}n(\bft)}  
q^{\frac{1}{2}n(\bft)+ \frac{2}{3}} (1 + u)^{\frac{3}{2}n(\bft)}.
\end{equation}
Utilizing this estimate, we obtain a new upper bound for the 
partition function of the $q$-state Potts model coupled to CTs
\begin{equation}
\Xi_N(\beta,\mu) \leq  q^{\frac{2}{3}} Z_N(\tilde{\mu}),
\end{equation}
where $\tilde{\mu}=\mu -\displaystyle\frac{3}{2} \ln\left( \displaystyle\frac{q+h}{q}\right) 
-\displaystyle\frac{1}{2}\ln q -\displaystyle\frac{3}{2}\ln(1 + u)$ and $Z_N(\tilde{\mu})$ is the 
partition function for pure CTs, defined in (\ref{et2-yamb}), on the cylinder $C_N$ with periodical
spatial boundary conditions and for the value of the cosmological constant $\tilde{\mu}$. 
Hence, 
the inequality 
\begin{equation}\label{eq2.35}
\mu > \displaystyle\frac{3}{2} \ln\left( q + e^\beta -1\right) + \ln 2
-\ln q +\displaystyle\frac{3}{2}\ln\left(1 + (q^{2/3} -1)\displaystyle\frac{e^\beta -1}{q+e^\beta -1} \right) 
\end{equation}
provides a sufficient condition for subcriticality behavior of the $q$-state Potts model coupled to CTs. 
Utilizing the  High-T expansion for $q$-state Potts model we get to obtained a better 
approximation of the critical curve. 
Inequality (\ref{eq2.35}) proof the part (a) of Theorem \ref{theo_bounds_main2}.

Now, using the duality relation of Theorem \ref{theo_duality_main1} and Eq. (\ref{eq2.35}), we obtain a new condition for 
subcriticality  behavior of the Potts model coupled to dual CTs
\begin{equation}\label{new_upperbound_dual}
\begin{array}{ccl}
\mu^* & > & \displaystyle\frac{3}{2}\beta^*  + \ln2 +   \displaystyle\frac{3}{2}\ln\left(1 + 
\displaystyle\frac{q^{2/3} -1}{e^{\beta^*}} \right)
\end{array}
\end{equation}
Inequality (\ref{new_upperbound_dual}) proof the part (b) of Theorem \ref{theo_bounds_main2}. This conclude the 
proof of Theorem \ref{theo_bounds_main2} because the same approach on dual triangulations does not improve the curves obtained.

\end{proof}

\section{$q=2$ (Ising) system}\label{conn_ising}
In this section we only consider Ising model on dual causal triangulations in order to compare our results with 
the previous results about the Ising model coupled dual causal triangulations present in \cite{HeAnYuZo:2013}. That early
work utilize the transfer matrix method in order to provide a curve $\mu^*=\psi(\beta^*)$ 
(blue line in Figure \ref{fig5.5}) that satisfies 
$$\frac{d\psi}{d \beta^*}(0^+)=0,$$ 
This property is in concordance with our result because the critical curve satisfied the same property, see Corollary 
\ref{asympt_beh}.  
\begin{figure}[t!]
\begin{center}
\includegraphics[width=9cm]{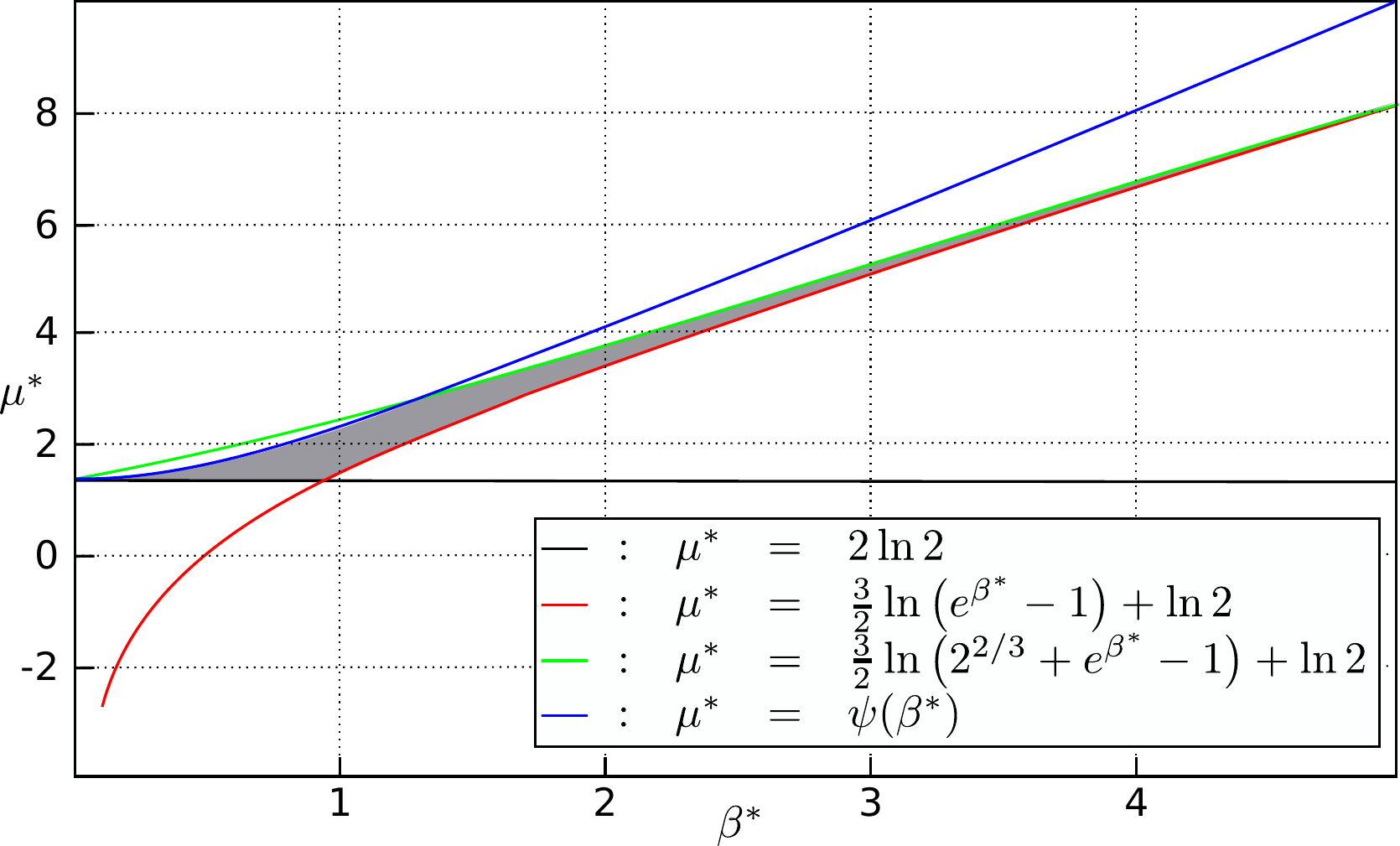}
\end{center}
\caption{Region where the critical curve of the Ising model coupled to dual CDTs can be located. Curves green, black and red 
are establish in Theorem \ref{theo_bounds_main2}, and the curve blue was compute in \cite{HeAnYuZo:2013}.}
\label{fig5.5}
\end{figure}
Now, we define the  functions 
$$\varphi_{inf}(\beta^*)=\max\left\{ 2\ln2,\frac{3}{2}\ln\left(e^{\beta^*}-1 \right) +\ln2 \right\},$$
and
$$\varphi_{sup}(\beta^*)=\min\left\{\psi(\beta^*),\frac{3}{2}\ln\left( 2^{2/3}+e^{\beta^*}-1 \right) +\ln2 \right\}.$$
Denoting by $\gamma_c^{I}$ the critical curve of the annealed Ising model, functions $\varphi_{inf}$ and $\varphi_{sup}$
provide a lower and an upper bound for the critical curve of the model, that is, 
$\varphi_{inf}(\beta)\leq \gamma_c^{I}(\beta)\leq \varphi_{sup}(\beta)$,
for any $\beta>0$. The upper bound $\varphi_{sup}$ improve bonds present in \cite{HeAnYuZo:2013}, \cite{cerda}, 
and confirm the expected behavior for the critical curve described in \cite{Ambjorn:1999gi}. 
Furthermore, in \cite{Ambjorn:1999gi} authors show numerical evidence that the model have a phase transition. 
Denoting by $\beta_c^{I}$ and $\beta_c^{*,I}$ the critical value of the coupled Ising model and its dual, 
respectively, and employing the duality relation of Theorem \ref{theo_duality_main1}, we have that 
\begin{equation}
\beta_c^{I} < \displaystyle\frac{1}{2}\ln(1+\sqrt{2})< \beta_c^{*,I}.
\end{equation}
Note that, this upper bound for the critical value of an Ising model coupled to causal triangulations do 
not improve the value $\beta_c^{I}\approx 0.2522$ computed in \cite{Ambjorn:1999gi} utilizing Monte Carlo simulation.

Finally, we  have that the free energy satisfy the following inequality
$$ \ln \Lambda\left(\mu^* -\varphi_{inf}(\beta^*) +\ln2\right) <
\lim_{N\to\infty} \frac{1}{N}\ln \Xi_N(\beta^*,\mu^*) < \ln \Lambda\left(\mu^* -\varphi_{sup}(\beta^*)+\ln2 \right).$$

\subsection*{Acknowledgements.} 
We would like to thank Prof. Y. Suhov and A. Yambartsev  for comments on a preliminary version of this article and for very valuable 
discussions and his encouragement. This work was supported by FAPESP, projects 2012/04372-7, 2013/06179-2 and 
2014/18810-1. 
Further, the author thanks the IME at the University of S\~ao Paulo for warm  hospitality.


\providecommand{\href}[2]{#2}


\begin{thebibliography}{10}

\bibitem{Ambjorn:1997di}
J.~Ambj{\o}rn, B.~Durhuus, and T.~Jonsson, {\em Quantum geometry. {A}
  statistical field theory approach}.
\newblock No.~1 in Cambridge Monogr. Math. Phys.,. Cambridge University Press,
  Cambridge, UK,
1997.
\newblock

\bibitem{Tutte1962a}
W.~T. Tutte, ``A census of planar triangulations,'' {\em Can. J. Math.} {\bf
  14} (1962) 21--38.

\bibitem{Tutte1963}
W.~T. Tutte, ``A census of planar maps,'' {\em Can. J. Math.} {\bf 15} (1963)
  249--271.

\bibitem{DiFrancesco:1993nw}
P.~Di~Francesco, P.~H. Ginsparg, and J.~Zinn-Justin, ``2-d gravity and random
  matrices,'' {\em Phys. Rept.} {\bf 254} (1995) 1--133,
\href{http://www.arXiv.org/abs/hep-th/9306153}{{\tt hep-th/9306153}}.

\bibitem{Schaeffer1997}
G.~Schaeffer, ``Bijective census and random generation of {E}ulerian planar
  maps with prescribed vertex degrees,'' {\em Elec. J. Comb.} {\bf 4} (1997)
  R20.

\bibitem{bouttier-2002-645}
J.~Bouttier, P.~{Di Francesco}, and E.~Guitter, ``Census of planar maps: {F}rom
  the one-matrix model solution to a combinatorial proof,'' {\em Nucl. Phys.}
  {\bf B645} (2002) 477, \href{http://www.arXiv.org/abs/cond-mat/0207682}{{\tt
  cond-mat/0207682}}.

\bibitem{Angel:2002ta}
O.~Angel and O.~Schramm, ``Uniform infinite planar triangulations,'' {\em Comm.
  Math. Phys.} {\bf 241} (2003) 191--213,
\href{http://www.arXiv.org/abs/math/0207153}{{\tt math/0207153}}.


\bibitem{Kazakov:1986hu}
V.~A. Kazakov, ``Ising model on a dynamical planar random lattice: {E}xact
  solution,'' {\em Phys. Lett.} {\bf A119} (1986)
140--144.

\bibitem{Boulatov:1986sb}
D.~V. Boulatov and V.~A. Kazakov, ``The {I}sing model on random planar lattice:
{T}he structure of phase transition and the exact critical exponents,'' {\em
Phys. Lett.} {\bf 186B} (1987)
379.


\bibitem{Onsager}
L.~Onsager, ``Crystal statistics. i. a two-dimensional model with an
  order-disorder transition,'' {\em Phys. Rev.} {\bf 65} (1944) 117--149.



\bibitem{Ambjorn:1998xu}
J.~Ambj{\o}rn and R.~Loll, ``Non-perturbative {L}orentzian quantum gravity,
  causality and topology change,'' {\em Nucl. Phys.} {\bf B536} (1998)
  407--434,
\href{http://www.arXiv.org/abs/hep-th/9805108}{{\tt hep-th/9805108}}.

\bibitem{Durhuus:2009sm}
B.~Durhuus, T.~Jonsson, and J.~F. Wheater, ``On the spectral dimension of
  causal triangulations,'' {\em J. Stat. Phys.} {\bf 139} (2010) 859--881,
\href{http://www.arXiv.org/abs/0908.3643}{{\tt 0908.3643}}.

\bibitem{SYZ1}
V.~Sisko, A.~Yambartsev, and S.~Zohren, ``A note on weak convergence results
  for uniform infinite causal triangulations,'' {\em Markov Proc. Related
  Fields} (2013).

\bibitem{Ambjorn:1999gi}
J.~Ambj{\o}rn, K.~N. Anagnostopoulos, and R.~Loll, ``A new perspective on
matter coupling in 2d quantum gravity,'' {\em Phys. Rev.} {\bf D60} (1999)
104035,
\href{http://www.arXiv.org/abs/hep-th/9904012}{{\tt hep-th/9904012}}.

\bibitem{Benedetti:2006rv}
D.~Benedetti and R.~Loll, ``Quantum gravity and matter: {C}ounting graphs on
causal dynamical triangulations,'' {\em Gen.Rel.Grav.} {\bf 39} (2007)
863--898, \href{http://www.arXiv.org/abs/gr-qc/0611075}{{\tt gr-qc/0611075}}.

\bibitem{Ambjorn:2008jg}
J.~Ambj{\o}rn, K.~N. Anagnostopoulos, R.~Loll, and I.~Pushkina, ``Shaken, but
not stirred - {P}otts model coupled to quantum gravity,'' {\em Preprint} (2008)
\href{http://www.arXiv.org/abs/0806.3506}{{\tt 0806.3506}}.

\bibitem{Ambjorn:2007jm}
J.~Ambj{\o}rn, R.~Loll, W.~Westra, and S.~Zohren, ``Putting a cap on causality
  violations in {CDT},'' {\em JHEP} {\bf 12} (2007) 017,
\href{http://www.arXiv.org/abs/arXiv:0709.2784 [gr-qc]}{{\tt arXiv:0709.2784
  [gr-qc]}}.




  



\bibitem{anatoli}
M.~Krikun and A.~Yambartsev, ``Phase transition for the {I}sing model on the
critical {L}orentzian triangulation,'' {\em Journal of Statistical Physics}, v. 148, p. 422-439, 2012.
\href{http://www.arXiv.org/abs/0810.2182}{{\tt 0810.2182}}.


\bibitem{MYZ2001}
V.~Malyshev, A.~Yambartsev, and A.~Zamyatin, ``Two-dimensional {L}orentzian
models,'' {\em Moscow Mathematical Journal} {\bf 1} (2001), no.~2, 1--18.


\bibitem{HeAnYuZo:2013}
 J. C. Hern\'andez, A. ~Yambartsev, Y. ~Suhov, S. ~Zohren, {\em Bounds on the critical line via transfer matrix
methods for an Ising model coupled to causal
dynamical triangulations}.  Journal of Mathematical Physics, {\bf v. 54}, p. 063301 (2013). 
\href{http://arxiv.org/pdf/1301.1483.pdf}{{\tt 1301.1483}}.

\bibitem{FK:1972}
C.M. Fortuin, R.W. ~Kasteleyn, {\em On the random-cluster model I. Introduction and relation to other models}. {\em
Physica.} {\bf 57},  536--564 (1972).

\bibitem{Ambjorn:2006}
J. Ambj{\o}rn, J. ~Jurkiewics, {\em The universe from scratch}.  {\em Contemporary Physics} {\bf 47}, 103-117 (2006).

\bibitem{Edwars_Sokal}
R. G. Edwards,  A. D. Sokal, {\em Generalization of the Fortuin-Kasteleyn- Swendsen-Wang representation and 
Monte Carlo algorithm}. Phys. Rev. (3){\bf 38}, p. 2009-2012 (1988).

\bibitem{Beffara_Duminil}
Beffara V.,  Duminil-Copin  H.: {\em The  self-dual point of the two- dimensional random-cluster 
model is critical for $q\geq 1$}. Probab. Theory Relat. Fields {\bf 153}, p. 511-542 (2012).

\bibitem{cerda}
Cerda-Hern\'andez  J.: {\em Critical region for an Ising model coupled to causal
dynamical triangulations}.
\href{http://http://arxiv.org/pdf/1402.3251.pdf}{{\tt 1402.3251}}.

\bibitem{Grimmett:1995}
Grimmett, G. R.: {\em The stochastic random-cluster process and the uniqueness of
random-cluster measures}.  {\em Ann. Probab.}. {\bf v. 23}(4), p. 1461--1510 (1995).

\bibitem{Grimmett:2006}
Grimmett, G. R.: {\em The random-cluster model}. Springer, Berlin (2006)

\bibitem{Napolitano:2015}
Napolitano, G. M. and Turova, T. {\em The Ising model on the random planar causal
triangulation: bounds on the critical line and magnetization properties}.(2015) 
\href{http://http://arxiv.org/pdf/1504.03828.pdf}{{\tt 1504.03828}}.

\bibitem{Domb}
Domb, C. {\em Configurational studies of the Potts models}.
{\em J. Phys.} A 7, p.1335 (1974).

\bibitem{Wu}
Wu, F. Y. {\em The Potts models}.
{\em Reviews of Modern Physics}, Vol.54(1), pp.235-268 (1982).

\bibitem{Baxter}
Baxter R. J. {\em  Exactly solved models in statistical mechanics} (1982).

\bibitem{Percol2015}
J. C. Hern\'andez, A. ~Yambartsev,  S. ~Zohren, {\em On the critical probability of percolation on random causal triangulations},
to appear {\em BJPS310}.


\end{thebibliography}
\end{document}